\documentclass[11pt]{article}

\usepackage{macros}

\newcommand{\cQ}{\mathcal{Q}}
\newcommand{\cP}{\mathcal{P}}

\newcommand{\LDLR}[3]{\mathsf{Adv}_{\leq #1}(#2, #3)}

\begin{document}
\title{The Quasi-Polynomial Low-Degree Conjecture is False}

\author{Rares-Darius Buhai\thanks{ETH Zurich. Email: \texttt{rares.buhai@inf.ethz.ch}.} \and Jun-Ting Hsieh\thanks{Carnegie Mellon University. Email: \texttt{juntingh@cs.cmu.edu}. Supported by NSF CAREER Award \#2047933.} 
\and Aayush Jain \thanks{Carnegie Mellon University. Email: \texttt{aayushja@cs.cmu.edu}. Supported by NSF CAREER Award \#2441647, and a Google Research Scholar Award.}\and Pravesh K. Kothari\thanks{Princeton University. Email: \texttt{kothari@cs.princeton.edu}. Supported by NSF CAREER Award \#2047933, NSF \#2211971, an Alfred P. Sloan Fellowship, and a Google Research Scholar Award.}}

\date{\today}

\maketitle
\begin{abstract}

There is a growing body of work on proving hardness results for average-case estimation problems by bounding the \emph{low-degree advantage} (LDA) --- a quantitative estimate of the closeness of low-degree moments --- between a \emph{null} distribution and a related \emph{planted} distribution. Such hardness results are now ubiquitous not only for foundational average-case problems but also central questions in statistics and cryptography.
This line of work is supported by the \emph{low-degree conjecture} of Hopkins~\cite{Hopkins18}, which postulates that a vanishing degree-$D$ LDA implies the absence of any noise-tolerant distinguishing algorithm with runtime $n^{\widetilde{O}(D)}$ whenever 1) the null distribution is product on $\{0,1\}^{\binom{n}{k}}$, and 2) the planted distribution is permutation invariant, that is, invariant under any relabeling $[n] \rightarrow [n]$.

In this paper, we disprove this conjecture. Specifically, we show that for any fixed $\varepsilon>0$ and $k\geq 2$, there is a permutation-invariant planted distribution on $\{0,1\}^{\binom{n}{k}}$ that has a vanishing degree-$n^{1-O(\varepsilon)}$ LDA with respect to the uniform distribution on $\{0,1\}^{\binom{n}{k}}$, yet the corresponding $\varepsilon$-noisy distinguishing problem can be solved in $n^{O(\log^{1/(k-1)}(n))}$ time. Our construction relies on algorithms for list-decoding for noisy polynomial interpolation in the high-error regime.

We also give another construction of a pair of planted and (non-product) null distributions on $\R^{n \times n}$ with a vanishing $n^{\Omega(1)}$-degree LDA while the largest eigenvalue serves as an efficient noise-tolerant distinguisher. 

Our results suggest that while a vanishing LDA may still be interpreted as evidence of hardness, developing a theory of average-case complexity based on such heuristics requires a more careful approach.

\end{abstract}

\thispagestyle{empty}
\setcounter{page}{0}

\newpage

\section{Introduction}
\newcommand{\Planted}{P}
\newcommand{\Null}{Q}

Central algorithmic challenges in a wide range of areas, from statistical estimation to cryptography, can be modeled as statistical \emph{signal detection} and \emph{recovery} problems.
In such problems, one must distinguish between an input drawn from a \emph{null} distribution and one drawn from a distribution with a \emph{planted} signal, ideally at the lowest possible signal strength.
The key research question is whether efficient distinguishers require a higher signal strength (the \emph{algorithmic threshold}) than inefficient ones (the \emph{statistical threshold}), a gap known as the \emph{information-computation gap}. Several foundational problems (e.g., planted clique, spiked Wigner and tensor models) are conjectured to exhibit information-computation gaps.
The modern research area of average-case complexity has developed a rich toolkit to identify such gaps. In addition to algorithm design and statistical estimation, this research direction has also provided a principled approach for new hardness assumptions underlying the security of cryptographic protocols~\cite{FOCS:MST03,STOC:AppBarWig10,C:AJLMS19,EC:JLMS19,EC:BBKK18,TCC:LomVai17,ABIKN23,BKR23,BJRZ24}. 

How should we build a rigorous theory of such  information-computation gaps? One strategy is to build a web of reductions starting from a few natural assumptions, paralleling fine-grained complexity (cf.~\cite{Williams18}). While achieving a fair amount of success in recent years \cite{BR13,HWX15,BBH18,BB19}, this approach is still limited to a restricted subset of problems.
This is because of the difficulty in designing reductions that transform a distribution on instances of one problem into the specific target distribution of another problem.\footnote{To highlight this difficulty, note that we do not know how to reduce refuting random $4$-SAT formulas with $n^{1+\Omega(1)}$-clauses to refuting random $3$-SAT formulas with $n^{1+\Omega(1)}$-clauses.}

\paragraph{Hardness against restricted algorithms} 
Most of the evidence for information-computation gaps has come from lower bounds against restricted families of algorithms (somewhat resembling lower bounds against weak circuit families).
A long sequence of works has focused on lower bounds against specific classes of spectral methods~\cite{MRZ16}, Markov chains~\cite{Jerrum92,CMZ25}, and convex programming hierarchies like the sum-of-squares (SoS) semidefinite programming relaxations.
In this context, the discovery of the \emph{pseudo-calibration} approach \cite{BHKKMP16} provided a heuristic that connects lower bounds against the sum-of-squares hierarchy to the \emph{low-degree advantage} (LDA) between a planted and a null distribution.

\begin{definition}[Low-Degree Advantage] \label{def:LDLR-intro}
The degree-$D$ advantage between two probability distributions $\Planted$ and $\Null$ on $\{0,1\}^N$ is defined as:
\begin{equation*}
    \LDLR{D}{\Planted}{\Null} \coloneqq \max_{f:\text{ deg-}D\text{ polynomial}} \frac{\E_{\Planted} f - \E_{\Null} f}{\sqrt{ \Var_{\Null}\, f }} \; .
\end{equation*}
\end{definition}

The low-degree advantage can be expressed in terms of the closeness between degree-$D$ moments of the two distributions.
It can also be interpreted as the degree-$D$ truncation of the likelihood ratio that captures the information-theoretic limits of statistical hypothesis testing (called \emph{low-degree likelihood ratio} (LDLR); see the survey \cite{KWB19}).
Unlike lower bounds against SoS relaxations, computing the LDA is often tractable in many settings.
The \emph{pseudo-calibration} conjecture~\cite{HKPRSS17} suggests that a vanishing LDA implies SoS lower bounds for an appropriately defined class of problems.
Indeed, starting with \cite{BHKKMP16}, a number of SoS lower bounds rely on the \emph{pseudo-calibration} technique where a lower bound witness is constructed from a pair of planted and null distributions, and showing that the LDA vanishes is a necessary first step~\cite{BHKKMP16,HKPRSS17,MRX20,GJJPR20,HK22} (though some recent works prove lower bounds based on planted distributions with non-vanishing LDA \cite{JPRTX21,JPRX23,KPX24}).

Starting with \cite{HS17}, vanishing LDA itself has been used as evidence of average-case hardness (see the survey \cite{KWB19}). Specifically, the \emph{low-degree heuristic} suggests that a distinguishing problem on (say) $n\times n$ matrix inputs is hard for $n^{o(D)}$ time algorithms if the degree-$D$ LDA asymptotically vanishes.
A flurry of follow-up work \cite{KWB19,GJW20,BH22,Wein22,Wein23,MW24,GJW24,Kunisky24,LG24,LZZ24,DMW25,li2025algorithmic,ding2025low,MW25,HM25}, not just in algorithm design but also statistics and cryptography, has used this heuristic to ascertain optimality of algorithms for new average-case problems or as supporting evidence for new hardness assumptions.
Further strands of research have developed analogs of the method for average-case estimation (as opposed to distinguishing) problems (e.g., \cite{SW22,MWZ23,MR4849263-Kunisky,LG24,MW24}).

\paragraph{Does vanishing LDA imply hardness?} The deluge of applications of the LDA method for hardness strongly motivates the investigation of whether and when (i.e., for what problems) a vanishing LDA predicts hardness. We know it doesn't, in general. Indeed, distinguishing between random linear equations vs those with a solution on $\mathbb{F}_2$ has $0$ LDA even at $\Omega(n)$-degree, but Gaussian elimination solves the problem in polynomial time. However, Gaussian elimination is brittle and fails even with a small amount of random noise (for e.g., corrupting a $o_n(1)$ fraction of equations; see \cite{DK22,ZSWB22} for similar examples).
In contrast, ``algorithms based on low-degree polynomials'' appear, informally, to tolerate such noise.
This led to the hypothesis that a vanishing LDA may imply failure of \emph{noise-tolerant} algorithms (analogous to the statistical query framework ~\cite{DBLP:conf/stoc/Kearns93,DBLP:conf/stoc/FeldmanGRVX13}), at least for problems with sufficient ``symmetry''.

\paragraph{The Low-Degree Conjecture}  Hopkins~\cite{Hopkins18} formulated a concrete hypothesis in his Ph.D. thesis that applies to all planted vs null distinguishing problems where the null and planted distributions are supported on $\{0,1\}^{n \choose k}$ --- viewed as $k$-tensors with all one-dimensional slices of size $n$ --- such that 1) the null is a product distribution on $\{0,1\}^{n \choose k}$ and 2) the planted distribution is permutation invariant. 
The \emph{low-degree conjecture}~\cite{Hopkins18} postulates that a vanishing degree-$D$ LDA implies the absence of an $n^{o(D)}$-time \emph{noise-tolerant} distinguishing algorithm whenever the two distributions satisfy the assumptions listed above.

Conjecture 2.2.4 of \cite{Hopkins18} formally applies to the special case of $D \sim \log^{1+\delta} n$ for $\delta>0$ with an informal general version appearing in Hypothesis 2.1.5 and a formal one inline in the discussion on Page 34. Several subsequent works have relied on the conjecture for all $D$.

\begin{conjecture}[The Low-Degree Conjecture, Hypothesis 2.1.5 and Conj 2.2.4 in \cite{Hopkins18}, Conj 2.1 in \cite{MR4345126-Ding21}, Conj 1.5 in \cite{MR4760356-Ding24}] \label{conj:low-deg-conj}
Fix $k \in \N$. Let $Q_n$ be the uniform distribution on $\{0,1\}^{n \choose k}$. Let $P_n$ be a distribution on $\{0,1\}^{n \choose k}$ that is invariant under the natural relabeling action of $S_n$.
If\linebreak ${\LDLR{D}{P_n}{Q_n} = O(1)}$, then for every fixed $\epsilon>0$, there is no $n^{D/\polylog(n)}$-time algorithm (for some $\polylog(n)$) that distinguishes between a sample from $T_{\epsilon} P_n$ and $Q_n$ with probability $1-o(1)$.
Here, $T_{\epsilon} P_n$ is the distribution obtained by drawing a sample from $P_n$ and replacing every coordinate with a uniformly random bit with probability $\epsilon$ independently. 
\end{conjecture}

Here, $S_n$-invariance means that for any permutation $\sigma \in S_n$, the distribution induced by mapping $M[i_1,\ldots i_k]$ to $M[\sigma(i_1),\ldots, \sigma(i_k)]$ is identical to the distribution of $M$. A large symmetry group is intended to preclude algorithms that try to exploit the presence of a small collection of special rows/columns, while noise tolerance is supposed to rule out ``algebraic" algorithms that are intuitively thought to be brittle (e.g., Gaussian elimination and lattice basis reduction).

The conditions of $S_n$-invariance and the null being product\footnote{The $S_n$-invariance forces any product null to have essentially identically distributed entries.} 
may appear restrictive, but they are satisfied by a host of well-studied distinguishing problems, including planted clique/dense subgraphs~\cite{BHKKMP16}, community detection~\cite{ding2025low,LZZ24,DBLP:journals/corr/abs-2502-14407}, and sparse PCA. Indeed, the intuitions for the truth of the conjecture arose from studying such problems.

To the best of our knowledge, \Cref{conj:low-deg-conj} imposes the most stringently formulated conditions on the pair of distributions for a vanishing LDA to imply hardness. In fact, substantial research effort has focused on \emph{expanding} the theory via variants of the conjecture suggesting that a vanishing LDA implies hardness even when the assumptions on $P_n$ and $Q_n$ in \Cref{conj:low-deg-conj} are not precisely met. A few examples include Conjecture 1.6 in~\cite{MR4345126-Ding21}, Conjecture 2.3 in \cite{moitra2023precise}, Conjecture 2 in~\cite{arpino2023statistical}, Conjecture 1.5 in~\cite{MR4760356-Ding24}, Conjecture 1.4 in~\cite{Kunisky24}, Conjecture 2.2 in~\cite{li2025algorithmic}, and Conjecture 1.3 in~\cite{ding2025low}.

Over time, the low-degree conjecture has been applied to justify using vanishing LDA (and related notions) as evidence of computational hardness. An abridged list of applications includes planted clique~\cite{BHKKMP16}, dense subgraphs~\cite{HKPRSS17}, sparse PCA~\cite{HKPRSS17,MR4760356-Ding24,dKNS20}, sparse clustering~\cite{LWB22}, stochastic block model \cite{HS17,BBKMW21,LG24,Kunisky24,JKTZ23,LZZ24} graph matching~\cite{MWXY24,DDL23,CDGL24}, planted dense cycles~\cite{MWZ23}, detecting geometry in random graphs \cite{BB23,BB24}, spiked Wigner and Wishart models~\cite{HS17,DBLP:conf/innovations/BandeiraKW20,KWB19,BBKMW21,MW24,BDT24}, planted submatrix~\cite{SW22} and variants with multiple communities \cite{RSWY23,DBLP:journals/corr/abs-2306-06643}, tensor PCA~\cite{DBLP:conf/nips/Chood21},
planted dense subhypergraph \cite{DBLP:journals/corr/abs-2304-08135}, planted hyperloops~\cite{BKR23}, sparse regression~\cite{DBLP:conf/nips/BandeiraAHSWZ22}, group testing~\cite{DBLP:journals/corr/abs-2206-07640}, Gaussian mixture models~\cite{DBLP:conf/colt/BrennanBH0S21,DBLP:journals/corr/abs-2207-04600,lyu2022optimalestimationcomputationallimit}, Gaussian graphical models~\cite{DBLP:conf/colt/BrennanBH0S21}, learning truncated Gaussians~\cite{DBLP:conf/colt/DiakonikolasKPZ24}, non-planted optimization problems such maximum independent set in sparse graphs~\cite{GJW24,Wein22,DBLP:journals/corr/abs-2501-06427} and hypergraphs~\cite{DBLP:journals/corr/abs-2404-03842}, maximum clique in $G(n,1/2)$~\cite{Wein22}, $k$-SAT~\cite{DBLP:conf/focs/BreslerH21}, spin glass optimization problems~\cite{GJW24}, and perceptron models~\cite{DBLP:journals/corr/abs-2203-15667}.

\paragraph{Is the Low-Degree Conjecture true?} 
Given the growing applications of the low-degree method, \Cref{conj:low-deg-conj}, if true, presents the exciting possibility of building a unified and principled theory of average-case complexity, at least under the assumptions on $P_n$ and $Q_n$.

At present, proving the conjecture appears beyond the reach of existing techniques, even modulo standard assumptions in worst-case or average-case complexity theory. On the other hand, no counter-example to \Cref{conj:low-deg-conj} has been found so far.
Prior works have explored and established the role of noise tolerance and symmetry in the truth of the conjecture.
Holmgren and Wein \cite{HW21} observed that any efficient unique decoding algorithm for an error correcting code in $\mathbb{F}_2^n$ with a large dual distance implies a counter-example to the conjecture \emph{if one were to drop the permutation-invariance condition}.

They also refuted the version of the conjecture that Hopkins wrote in the setting where the domain $\Omega=\R$ and $Q_n$ is a standard Gaussian distribution, by exploiting the fact that one can encode a large amount of information in a single uncorrupted real number.
However, they observed that their technique no longer gives a counter-example if one demands noise-tolerance in a way that is more natural, in retrospect, when $Q_n$ is Gaussian. Specifically, they noted that the right analog of the Boolean noise operator in \Cref{conj:low-deg-conj} should be the Ornstein-Uhlenbeck operator that adds a small independent Gaussian to every entry, as opposed to corrupting only a small constant fraction of the entries as in the original proposal in \cite{Hopkins18}.

Similarly, the work of \cite{ZSWB22,DK22} shows algorithms based on lattice reductions that can solve problems in regimes where the LDA vanishes.
However, these algorithms, like Gaussian elimination, are not noise tolerant and fail under a vanishing amount of random noise.
The failure to disprove \Cref{conj:low-deg-conj} so far has served as an argument in favor of the conjecture.
In light of such attempts and the importance of the conjecture, refuting or gaining more evidence for it was pointed out as a major research direction in a recently concluded workshop on low-degree polynomial methods in average-case complexity~\cite{AIM24}.

In this work, we show that \Cref{conj:low-deg-conj} is false. We also give another example of a problem over matrices in $\R$, in which $Q_n$ is rotational invariant but not a product distribution, and the $n^{\Omega(1)}$-degree LDA asymptotically vanishes while the largest eigenvalue serves as a distinguisher.
We describe both results in detail below.

\subsection{Our Results}

Our first example satisfies all the conditions of \Cref{conj:low-deg-conj}, has $0$ $n^{1-O(\epsilon)}$-degree LDA, while a quasi-polynomial time algorithm succeeds in solving the distinguishing problem.

\begin{theorem}[\Cref{conj:low-deg-conj} is false; see \Cref{thm:reed-solomon}]
\label{thm:main-reed-solomon}
For every $\epsilon>0$ and integer $k \geq 2$, there is a distribution $P_n$ on $\{0,1\}^{n \choose k}$ that satisfies the conditions of \Cref{conj:low-deg-conj} such that $\LDLR{D}{P_n}{Q_n} = 0$ for $D=n^{1-O(\epsilon)}$ while there is a $n^{O(\log^{1/(k-1)} (n))}$-time distinguisher for $T_{\epsilon} P_n$ and $Q_n$ that succeeds with probability $1-o_n(1)$.
\end{theorem}
It is not hard to show (see \Cref{rem:Boolean-to-Gaussian}) that any counter-example for the Boolean setting (such as above) implies a similar counter-example for the case when $Q_n$ is the standard Gaussian distribution and $T_{\epsilon}$ is the Gaussian Ornstein-Uhlenbeck noise operator (suggested in the refined version of the low-degree conjecture for Gaussian $Q_n$ in \cite{HW21}). 

We note that a concurrent work \cite{concurrent2025} shows that \Cref{conj:low-deg-conj} is true for $k=1$ as in a vanishing degree-$O(\log n)$ LDA implies the failure of \emph{all} distinguishers for the corresponding noisy distinguishing problem. Thus, taken together, our results give a complete resolution of the Boolean alphabet case of \Cref{conj:low-deg-conj}.

\paragraph{Polynomial-time distinguisher for rectangular inputs}
The $S_n$-symmetry requirement makes \Cref{conj:low-deg-conj} quite restrictive.
First, it only applies to square symmetric matrices (or more generally, tensors with all slices of the same dimension).
Moreover, in the case of Boolean domain, any $S_n$-symmetric product distribution (for the null) is essentially a distribution over $n$-vertex undirected hypergraphs where each hyperedge is sampled i.i.d.

It is natural to postulate a generalization that applies to rectangular matrices (or tensors with slices of unequal dimensions). In such a case, the symmetry requirement must be reformulated (e.g., for a bipartite graph with left vertex set of size $m$ and right vertex set of size $n$, the relabeling should not send a left vertex to a right one).
This setting already arises in several applications of the low-degree heuristic, including spiked Wishart models~\cite{DBLP:conf/innovations/BandeiraKW20, dKNS20, MR4345126-Ding21, BDT24} and bipartite planted clique models~\cite{BKS23}.

In \Cref{rem:rectangular,thm:reed-solomon-partite}, we show that under the natural definition of symmetry in the rectangular setting,
our example yields a \emph{polynomial-time distinguisher}.
This formally refutes the heuristic that $\polylog(n)$-degree indistinguishability rules out polynomial-time (noise-tolerant) distinguishing algorithms, albeit not satisfying the $S_n$-symmetry required in \Cref{conj:low-deg-conj}.

\paragraph{Noisy polynomial interpolation} 


Our counter-example is based on the well-studied \emph{noisy polynomial interpolation} for list-decoding (modified to fit the specifications of \Cref{conj:low-deg-conj}):

\begin{center}\emph{Given $n$ evaluations of a univariate degree-$m$ polynomial $p$ on $F_q$ on random inputs corrupted independently with probability $1-1/n^{O(\epsilon)}$, find a list of polynomials that includes $p$ of size $\poly(n)$.}\end{center}

We note that such a problem arises naturally in many applications, including in cryptography (e.g., in analyzing the security of \emph{McEliece cryptosystems} \cite{McEliece1978,CouvreuL22,DBLP:conf/tcc/BoyleHW19,DaoJ24,mceliece-rm-codes,SS92-attack-mceliece-grs}). In fact, McEliece cryptosystems involve decoding randomly permuted noisy evaluations, similar to the permutations involved in our counter-example. We note that while the low-degree conjecture has not been used in the context of such cryptosystems so far, it has been invoked in the security analyses of other cryptographic protocols \cite{ABIKN23,BKR23,BJRZ24}.

\paragraph{Error correction under noise and permutation}
We will describe the construction below by formulating a general problem of constructing an efficient list-decodable error-correcting code where codewords are viewed as $k$-fold tensors in $n$ dimensions.
We are interested in tolerating an adversary that, in addition to the usual corruptions, can also apply an arbitrary permutation to relabel the coordinates of the tensor. Any efficient code that satisfies such a condition  immediately yields a counter-example to \Cref{conj:low-deg-conj}. This problem appears to be independently interesting and closely related to other works in coding theory, including the recent work on \emph{graph codes} \cite{alon2024graph}. We will show how noisy polynomial interpolation can give an efficient construction of such a permutation-resilient, list-decodable error-correcting code and thus obtain our counter-example.



\begin{definition}[Permutation-resilient, efficiently list-decodable codes] \label{def:permutation-resilient-code}
Let $E:\{0,1\}^m \rightarrow \{0,1\}^{n \choose k}$ be a (possibly randomized) encoding map where we view the codewords as symmetric tensors of order $k$. We say that $E$ is \emph{permutation-resilient, efficiently list-decodable} if, given $y = E(x)$ obtained by 1) flipping every entry of $y$ with probability $\epsilon$ independently, 2) applying the relabeling action of a uniformly random $\sigma \in S_n$ on $y$, one can efficiently construct a list of $\poly(n,k)$ messages guaranteed to contain $x$ with high probability over $E$, the corruptions and the permutations.
\end{definition}

We then make the following observation: 
\begin{observation}
\label{obs:obs1}
Suppose there is an encoding map as in \Cref{def:permutation-resilient-code} such that the distribution of $E(x)$ for a random $x$ is $D$-wise uniform and the list-decoding algorithm tolerates a constant rate $\epsilon>0$ of corruptions and runs in time $n^{o(D/\polylog(n))}$. Then, \Cref{conj:low-deg-conj} is false. 
\end{observation}

To see why, we choose $P_n$ by 1) choose a uniformly random permutation $\sigma$, 2) choose a uniformly random $x \in \{0,1\}^m$, and output $\sigma (E(x))$. Then, $P_n$ is clearly $S_n$-invariant and has a vanishing degree-$D$ LDA with respect to the uniform distribution $Q_n$. On the other hand, one can simply apply the list-decoding algorithm and note that such an algorithm must necessarily fail with high probability on $Q_n$ to obtain a distinguisher.

In \Cref{sec:reed-solomon}, we show how to construct such efficiently list-decodable permutation-resilient codes based on Reed-Solomon codes in its list decoding regime.  Our code is list-decodable in time $n^{O(\log n)}$ for $k=2$ and $n^{O(\log^{1/k-1}n)}$ for a general $k \in \N$. The vanishing LDA follows easily from the dual code having a large distance while our efficient distinguisher uses a high-error list-decoding algorithm (e.g.,~\cite{Sudan97,Guruswami-Sudan}).

\paragraph{Low-degree polynomials cannot exactly compute the eigenvalues} In the second part of this paper, we give a different construction of $P_n$ and $Q_n$ on $n \times n$ matrices in $\R$ such that $P_n, Q_n$ satisfy permutation invariance (in fact, even the stronger property of \emph{rotation} invariance) while the top eigenvalue of the input matrix serves as a polynomial-time distinguisher that succeeds with high probability even in an arguably natural noise model in the setting.

Our example in this case is also simple and is based on a carefully designed eigenvalue distribution of $n \times n$ matrices with eigenvectors being the columns of an independent and random orthogonal matrix. In our construction, we note that $Q_n$ is not a product distribution, and thus, this construction does not refute \Cref{conj:low-deg-conj}. However, as discussed before, there are many applications of the low-degree conjecture where the null distribution does not satisfy the product requirement (see e.g.,~\cite{Wein23,RSWY23,KVWX23,BB24}). We note that under $Q_n$, the correlation between any two (or a constant number of) entries of the matrix-valued random variable is $o_n(1)$.

\begin{theorem}[Informal \Cref{thm:eigenvalue}]
\label{thm:main-eigenvalue}
    There are rotational-invariant distributions $Q_n$ and $P_n$ over matrices $\R^{n \times n}$ such that
    $\LDLR{D}{P_n}{Q_n} \leq o(1)$ for $D = \wt{\Omega}(n^{1/3})$ while there is a $\poly(n)$-time algorithm that distinguishes between $Q_n$ and a noisy $P_n$.
\end{theorem}

This result goes against the conventional wisdom that spectral methods, at least of the simple kind that compute the largest eigenvalue, are ``captured by $\polylog(n)$-degree polynomials''.\footnote{In fact, several papers informally comment that ``$O(\log n)$-degree polynomials capture spectral methods''. See for e.g., \cite{GJW20,Wein23,DMW25,HM25}.}
This intuition is based on the fact that for any symmetric matrix $A \in \R^{n \times n}$, $\|A\|_2 \leq \tr(A^{2k})^{1/2k} \leq n^{1/2k} \|A\|_2$.
Thus, with $k = \omega(\log n)$, we have $n^{1/2k} \leq 1 + O(\frac{1}{k}\log n)$, and $\tr(A^{2k})$ (a degree-$2k$ polynomial of $A$) approximates the norm up to a $(1+o(1))$ factor.

Our two distributions in \Cref{thm:main-eigenvalue}, after a shift of eigenvalues, are over matrices of norm $1$ and $1+\lambda^*$ respectively, where $\lambda^*$ is chosen to be $\frac{1}{\poly(n)}$.
Thus, $\tr(A^{2k})$ fails to distinguish even if $k = n^c$ for some constant $c$.
This is our main intuition for \Cref{thm:main-eigenvalue}.

\begin{remark}[Noise Model]
    Despite the various applications of the low-degree conjecture in settings where $Q_n$ is not a product distribution, there is no precise formulation of the noise model under which noise-tolerant algorithms are conjectured to be ruled out.
    In \Cref{thm:main-eigenvalue}, we consider the noise model that adds a scaled copy of an independent draw from the null distribution.
    This aligns with the Ornstein-Uhlenbeck noise model in the setting where $Q_n$ is Gaussian.
    In our case, however, this noise changes both our planted and null distributions.
    Thus, we additionally prove that the LDA is vanishing even for the noisy versions of the null and planted distributions.
\end{remark}

Here, we give a brief overview of \Cref{thm:main-eigenvalue}.
The null distribution $Q_n$ is supported on negative semidefinite matrices.
Specifically, we sample $\lambda_1,\dots, \lambda_n$ independently from some distribution $\mu$ over $[-1,0]$, and output $U \diag(\lambda) U^\top$ where $U$ is a random rotation matrix.\footnote{In our proof, for convenience, we set $U$ to be a matrix with Gaussian entries instead. The difference is negligible.}
For the planted distribution $P_n$, we do the same for $\lambda_1,\dots,\lambda_{n-1}$, but set $\lambda_n = \lambda^* > 0$.
Intuitively, due to the rotational invariance, we only need to consider distinguishers that are symmetric functions of $\lambda$.
Suppose we set $\lambda^* = \frac{1}{\poly(n)}$, then low-degree functions of $\lambda$ should not be able to detect the small positive eigenvalue.

We also need to show that after adding noise to a matrix $M \sim P_n$, the matrix still has a positive eigenvalue.
Since $P_n, Q_n$ are not product distributions, there is no standard notion of a noise operator.
As a concrete example, we look at one natural definition: for a sample $M \sim P_n$, the noisy output is $N' = (1-\eps) M + \eps M_0$ for $M_0$ sampled from the null $Q_n$.

It suffices to show that for the eigenvector $u$ where $u^\top M u = \lambda^*$, we have $u^\top M' u = (1-\eps)\lambda^* + \eps \cdot u^\top M_0 u > 0$.
Since $u$ is a random vector, $u^\top M_0u$ will be concentrated around (a scaling of) $\tr(M_0)$.
This requires our matrices to be low rank.
Specifically, we require $M_0$ to have rank roughly $\wt{O}(\lambda^* n)$ --- i.e., the distribution $\mu$ outputs $0$ with probability $1-\wt{O}(\lambda^*)$.
See \Cref{sec:top-eigenvalue} for more details.

\subsection{Discussion}

In retrospect, it is perhaps unsurprising that a single heuristic, such as vanishing LDA, fails to characterize efficient algorithms, even when we impose additional requirements such as noise tolerance and symmetry. Still, the low-degree conjecture has resisted attempts at refutation since its introduction in 2017, despite the wide range of applications, especially in the last five years. 

\paragraph{What does a vanishing LDA mean for computational hardness?}
For well-studied problems such as finding planted cliques in random graphs or refuting random constraint satisfaction problems, all algorithmic efforts have failed to improve the best-known algorithms from more than two decades ago. Although one may still hesitate to conjecture\footnote{To paraphrase a famous line, \emph{algorithms find a way}.} computational hardness at the thresholds, the failure to find a better algorithm despite decades of effort is, by itself, perhaps as strong an evidence of hardness as any. But how should we interpret vanishing LDA based hardness for a \emph{new and relatively unexplored} problem, such as those that may arise in cryptography \cite{ABIKN23,BKR23,BJRZ24}? In such cases, our work suggests significant caution.  

\paragraph{Improving our counter-example}   One could improve our counter-example somewhat and find one where the distinguisher runs in polynomial as opposed to $n^{\log^{\delta}(n)}$ time for an arbitrarily small $\delta$ as in our current construction. A natural avenue for this is building an efficient permutation-invariant, list-decodable code with large dual distance as in \Cref{obs:obs1}.
As noted earlier, for the rectangular generalization of Hopkins' conjecture, our construction already yields a counter-example with a polynomial-time distinguisher (see~\Cref{rem:rectangular,thm:reed-solomon-partite}).

Relatedly, finding more examples of algorithmic techniques that circumvent \Cref{conj:low-deg-conj} is also an important research direction, as it suggests natural avenues for surpassing lower bounds via the LDA method for specific problems. 

\paragraph{Reformulating the conjecture?} It is natural to ask if certain additional natural conditions on the pair of distributions $P_n$ and $Q_n$ could lead to a potentially viable version of \Cref{conj:low-deg-conj} while still satisfied by well-studied average-case problems.

The low-degree conjecture (and the connections to other restricted algorithmic frameworks such as the overlap gap property~\cite{DBLP:journals/corr/abs-1904-07174}, statistical query model~\cite{DBLP:conf/stoc/Kearns93,DBLP:conf/stoc/FeldmanGRVX13}, and the Franz-Parisi criterion~\cite{DBLP:conf/nips/BandeiraAHSWZ22} from statistical physics) have fueled a recent excitement for a principled theory of average-case complexity based on such heuristics and connections \cite{DBLP:conf/colt/BrennanBH0S21}. We believe that the development of such a theory will benefit from precisely stated conjectures (such as Hopkins's conjecture), rigorous investigations of their truth, and rigorous characterizations of what algorithmic techniques their predictions apply to. 

\paragraph{Concrete implications of vanishing LDA?}
In general, it appears unlikely to us that we will be able to formulate a single tractable heuristic that captures all efficient algorithms.
However, it may still be feasible to rigorously characterize the class of algorithms that such heuristics can help rule out.
Such attempts will be valuable for both algorithm designers and cryptographers who seek provable hardness against restricted classes of algorithms.
So far, there is little work in this direction in the context of the vanishing LDA heuristic, and it was suggested as a major research direction in a recent workshop on the topic~\cite{AIM24}.
Notably, a concurrent work~\cite{concurrent2025} makes concrete progress on this front.

\section{Low-Degree Conjecture vs Noisy Polynomial Interpolation}
\label{sec:reed-solomon}

In this section, we prove \Cref{thm:main-reed-solomon}. We let $T_{\eps}$ denote the standard Boolean noise operator. For any distribution $\calD$ over $\{0,1\}^N$, $T_\eps \calD$ is the distribution where (1) we sample $X \sim \calD$, then (2) independently for each coordinate of $X$, we replace it with a uniform sample from $\{0,1\}$ with probability $\eps$.

\begin{theorem}
\label{thm:reed-solomon}
    Fix any integer $k\geq 2$ and small enough $\eps > 0$.
    Let $\cQ_n$ be the uniform distribution over symmetric $k$-tensors in $(\{0,1\}^n)^{\otimes k}$.
    Then, there exists an $S_n$-symmetric distribution $\cP_n$ over symmetric $k$-tensors in $(\{0,1\}^n)^{\otimes k}$ such that
    \begin{itemize}
        \item \textbf{Degree-$n^{1-O(\eps)}$ indistinguishable:}
        $\LDLR{D}{\cP_n}{\cQ_n} = 0$ for $D = n^{1-6\eps}$.
        \item \textbf{Distinguishing algorithm after noise:} there is an algorithm $\calA$ that runs in time $n^{O(\log^{1/(k-1)} n)}$ such that
        $\Pr_{M \sim \cQ_n}[\calA(M) = 0], \Pr_{M \sim T_\eps \cP_n}[\calA(M)=1] \geq 1- \exp(-n^{1-O(\epsilon)})$.
    \end{itemize}
\end{theorem}

\begin{remark}[Boolean counter-example translates to a Gaussian counter-example] \label{rem:Boolean-to-Gaussian}
One can naturally extend any Boolean counter-example to the setting where $Q_n$ is the distribution of a symmetric tensor with independent Gaussian entries. To do this, we sample a $k$-tensor $T$ from $P_n$ or $Q_n$, treat it as a tensor with $\pm 1$-entries (instead of $0$-$1$), and let $T'$ be obtained by multiplying each entry of $T$ with the absolute value of an independent standard Gaussian. We thus get a pair of distributions on $\R^{n\choose k}$, and further, the null distribution is that of a symmetric tensor with independent Gaussian entries. The proof that LDA vanishes directly extends to this variant. We can also extend the distinguishing algorithm by first taking the entry-wise sign of the input tensor to obtain a Boolean tensor and then applying the algorithm for the Boolean case. The noise-tolerance analysis extends naturally by observing that the classical Sheppard's Lemma implies that a noise rate of $\epsilon$ for the Gaussian Ornstein-Uhlenbeck noise operator translates into a noise rate of $O( \sqrt{\epsilon})$ for the Boolean setting obtained by taking entry-wise signs.
\end{remark}

\begin{remark}[Rectangular version of Hopkins' conjecture and polynomial-time distinguisher] \label{rem:rectangular}
Suppose we have rectangular tensors in $\{0,1\}^{n_1 \times n_2 \times \cdots \times n_k}$.
Then, we need to consider $(S_{n_1} \times S_{n_2} \times \cdots \times S_{n_k})$-symmetry --- that is, invariance under independent permutations of indices along each mode.
Such distributions can be viewed as distributions over $k$-partite hypergraphs.
In this case, we can construct two distributions that are degree-$n^{1-O(\eps)}$ indistinguishable but have a distinguishing algorithm with $\poly(n)$ runtime.
\end{remark}

\begin{theorem}
\label{thm:reed-solomon-partite}
    Fix any small enough $\eps > 0$.
    Let $\cQ_n$ be the uniform distribution over $3$-tensors in $\{0,1\}^{\sqrt{\log n} \times \sqrt{\log n} \times n}$.
    Then, there exists an $(S_{\sqrt{\log n}} \times S_{\sqrt{\log n}} \times S_n)$-symmetric distribution $\cP_n$ over $3$-tensors in $\{0,1\}^{\sqrt{\log n} \times \sqrt{\log n} \times n}$ such that
    \begin{itemize}
        \item \textbf{Degree-$n^{1-O(\eps)}$ indistinguishable:}
        $\LDLR{D}{\cP_n}{\cQ_n} = 0$ for $D = n^{1-6\eps}$.
        \item \textbf{Distinguishing algorithm after noise:} there is an algorithm $\calA$ that runs in time $\poly(n)$ such that
        $\Pr_{M \sim \cQ_n}[\calA(M) = 0], \Pr_{M \sim T_\eps \cP_n}[\calA(M)=1] \geq 1- \exp(-n^{1-O(\epsilon)})$.
    \end{itemize}
\end{theorem}

The proof of \Cref{thm:reed-solomon} immediately implies \Cref{thm:reed-solomon-partite}, thus we will omit a detailed proof.
The distinguishing algorithm runs in polynomial time because we can exhaustively search over $(\sqrt{\log n})!^2 = n^{o(1)}$ permutations of the first two modes.
In fact, the distributions in \Cref{thm:reed-solomon} (for $k=3$) can be viewed as taking the distributions in \Cref{thm:reed-solomon-partite} and then padding random bits (along with random permutations) to form $S_n$-symmetric tensors in $(\{0,1\}^n)^{\otimes 3}$.
The algorithm then needs to search over $n^{O(\sqrt{\log n})}$ indices, hence the runtime as stated in \Cref{thm:reed-solomon}.

\subsection{Preliminaries on Reed-Solomon Codes}

The Reed-Solomon code \cite{RS60} is a family of error-correcting codes obtained by evaluating low-degree polynomials over a large field.

\begin{definition}[Reed-Solomon code]
    Let $\F_q$ be a large field with $q \geq n$ a prime power.
    Given $m\in \N$ and distinct $\alpha_1, \dots, \alpha_n \in \F_q$, the Reed-Solomon code is defined as
    \begin{align*}
        \braces*{ (p(\alpha_1), p(\alpha_2),\dots, p(\alpha_n)) \in \F_q^n: p \text{ is a polynomial over $\F_q$ of degree $< m$}} \mper
    \end{align*}
    The natural way to encode a message $x= (x_0, x_1,\dots,x_{m-1}) \in \F_q^m$ is by setting $p_x(\alpha) = \sum_{j=0}^{m-1} x_j \alpha^j$.
\end{definition}

It is well-known that the Reed-Solomon code is $(m-1)$-wise independent:

\begin{fact}[$(m-1$)-wise independence of codeword distribution, see Proposition 4.2 in \cite{HW21}]
\label{fact:indep-wise}
For $q \geq n$ a prime power, let $x_0, \ldots, x_{m-1} \stackrel{i.i.d.}{\sim} \operatorname{Unif}(\mathbb{F}_q)$.
Fix any distinct $\alpha_1, \ldots, \alpha_n \in \mathbb{F}_q$, and define $\beta_1, \ldots, \beta_n \in \mathbb{F}_q$ as $\beta_i = \sum_{j=0}^{m-1} x_j \alpha_i^j$.
Then the marginal distribution over any $m-1$ $\beta_i$s is $\operatorname{Unif}(\mathbb{F}^{m-1})$.
\end{fact}

The key fact we will need is that Reed-Solomon codes are list-decodable.

\begin{fact}[Guruswami-Sudan list decoding, see Theorem 8 and Theorem 12 in \cite{Guruswami-Sudan}]
\label{fact:GS}
Given $n$ points $\{(x_i, y_i)\}_{i=1}^n$ with $x_i, y_i$ in a field $\mathbb{F}$ of cardinality at most $2^n$, for $t > \sqrt{nm}$ there exists an algorithm that runs in time $O(n^{15})$ and outputs a list of size at most $O(n^{15})$ of all polynomials $p$ of degree at most $m$ such that $y_i = p(x_i)$ for at least $t$ values $i \in [n]$.
\end{fact}

\subsection{Proof of \texorpdfstring{\Cref{thm:reed-solomon}}{Theorem~\ref{thm:reed-solomon}}}

We first define our distributions for matrices, i.e., the case $k=2$.

\begin{definition}[Null distribution $\cQ_n$] \label{def:null}
    The distribution $\cQ_n$ is the distribution over symmetric matrices $M \in \{0, 1\}^{n \times n}$ with entries $M_{i,j}$ with $i < j$ sampled i.i.d. from $\operatorname{Unif}(\{0, 1\})$.
\end{definition}

For $q \geq 2$ a power of two, we define $\operatorname{binary}_q: \mathbb{F}_q \to \{0, 1\}^{\log_2 q}$ to map the $i$-th element of $\mathbb{F}_q$ to a canonical binary representation of $i$.

\begin{definition}[Planted distribution $\cP_n$] \label{def:planted}
    For some $m \geq 2$ and $2 \leq q \leq 2^{\Omega(n)}$ a power of two, the distribution $\cP_n$ is sampled as follows:
    \begin{enumerate}[(1)]
        \item Sample $x_0, \ldots, x_{m-1} \stackrel{i.i.d}{\sim} \operatorname{Unif}(\mathbb{F}_q)$.
        \item Sample $\alpha_1, \ldots, \alpha_{\floor{n/2}} \stackrel{i.i.d.}{\sim} \operatorname{Unif}(\mathbb{F}_q)$ and define $\beta_1, \ldots, \beta_{\floor{n/2}} \in \mathbb{F}_q$ as $\beta_i = \sum_{j=0}^{m-1} x_j \alpha_i^j$ for all $i=1,\ldots,\floor{n/2}$.
        For any $\alpha_i$ whose value appears more than once among $\alpha_1, \ldots, \alpha_{\floor{n/2}}$, resample $\beta_i \sim \operatorname{Unif}(\mathbb{F}_q)$.
        \item Let the matrix $M_0 \in \{0,1\}^{2 \log_2q \times \floor{n / 2}}$ have:
        \begin{itemize}
            \item $M_0[1 : \log_2 q,\; j] = \operatorname{binary}_q(\alpha_{j})$, for all $j=1, \ldots, \floor{n/2}$,
            \item $M_0[\log_2 q + 1 : 2\log_2 q,\; j] = \operatorname{binary}_q(\beta_{j})$, for all $j=1, \ldots, \floor{n/2}$.
        \end{itemize}
        \item Let the symmetric matrix $M \in \{0, 1\}^{n \times n}$ have $M[1: 2\log_2 q, \; \ceil{n/2} + 1 : n] = M_0$, and let all its other entries $M_{i,j}$ with $i < j$ be sampled i.i.d. from $\operatorname{Unif}(\{0, 1\})$.
        \item Apply a random $S_n$-permutation to the matrix $M$ and return it.
    \end{enumerate}
\end{definition}

First, we prove that in the planted model the marginal distribution on any $m-1$ entries is uniform.
This implies automatically a low-degree lower bound of degree $m-1$.

\begin{lemma}
\label{lemma:rs-indist}
Let $q \geq n$ be a power of two.
For $M \sim \cP_n$, the marginal distribution on any $m-1$ entries in $\{M_{i,j} \mid i < j\}$ is $\operatorname{Unif}(\{0, 1\}^{m-1})$.
\end{lemma}
\begin{proof}
Let us condition on the $S_n$-permutation that is applied in the last step in \cref{def:planted}.
Then, for a fixed $S_n$-permutation, any entry that does not correspond to $\alpha$s and $\beta$s is sampled independently from $\operatorname{Unif}(\{0, 1\})$, so it suffices to only prove the result for entries corresponding to $\alpha$s and $\beta$s.
Let $a_1, \ldots, a_{m-1}$ be $m-1$ entries corresponding to $\alpha$s and $b_1, \ldots, b_{m-1}$ be $m-1$ entries corresponding to $\beta$s.
Then 
\[\Pr_{M \sim \cP_n}(a_1, \ldots, a_{m-1}, b_1, \ldots, b_{m-1}) = \Pr_{M \sim \cP_n}(b_1, \ldots, b_{m-1} \mid a_1, \ldots, a_{m-1}) \Pr_{M \sim \cP_n}(a_1, \ldots, a_{m-1})\,,\]
where $\Pr_{M \sim \cP_n}(a_1, \ldots, a_{m-1})$ is uniform by definition.

For the first term on the right-hand side, we prove the stronger fact that $\Pr_{M \sim \cP_n}(b_1, \ldots, b_{m-1} \mid \alpha_1, \ldots, \alpha_{n})$ is uniform for any $\alpha_1, \ldots, \alpha_n$.
By \cref{fact:indep-wise}, for distinct $\alpha$s, it follows that the distribution over any $m-1$ $\beta$s is uniform.
On the other hand, recall that if some some $\alpha$s happen to be identical, the corresponding $\beta$s are resampled uniformly.
Let $S \subseteq \{b_1, \ldots, b_{m-1}\}$ be the subset of $b_1, \ldots, b_{m-1}$ for which the corresponding $\beta$s are resampled uniformly, and let $S' = \{b_1, \ldots, b_{m-1}\} \setminus S_b$.
Then
\begin{align*}
\Pr_{M \sim \cP_n}(b_1, \ldots, b_{m-1} \mid \alpha_1, \ldots, \alpha_{n})
&= \Pr_{M \sim \cP_n}(S \mid \alpha_1, \ldots, \alpha_{n}, S') \Pr_{M \sim \cP_n}(S' \mid \alpha_1, \ldots, \alpha_{n})\,,
\end{align*}
where both terms are uniform based on the discussion above.
This concludes the proof.
\end{proof}

Second, we prove that there is quasi-polynomial time algorithm that distinguishes the null and planted distributions with high probability, even when noise is applied to the planted distribution.

\begin{lemma}
\label{lemma:rs-dist}
For $q = \Theta(n)$ with $q \geq n$ a power of two, $\epsilon > 0$ a small enough constant, and $m \leq n^{1-6\epsilon}$, there exists an algorithm $\mathcal{A}$ that, given as input a symmetric matrix $M \in \{0, 1\}^{n \times n}$, runs in time $n^{O(\log_2 q)}$ and satisfies
\[\Pr_{M \sim \cQ_n}(\mathcal{A}(M) = 0) \geq 1-\exp(-n^{1-O(\epsilon)})\,,\,\quad \Pr_{M \sim T_\epsilon \cP_n}(\mathcal{A}(M) = 1) \geq 1-\exp(-n^{1-O(\epsilon)})\,.\]
\end{lemma}
\begin{proof}
The algorithm is:
\begin{enumerate}
    \item Guess $2 \log_2 q$ distinct ordered indices $i_1, \ldots, i_{2 \log_2 q} \in [n]$.
    \item For each $j \in S = \{1, \ldots, n\} \setminus \{i_1, \ldots, i_{2 \log_2 q}\}$, let
    \[\alpha_j = \operatorname{binary}_q^{-1}(M[(i_1, \ldots, i_{\log_2 q}), j])\,,\]
    \[\beta_j = \operatorname{binary}_q^{-1}(M[(i_{\log_2 q + 1}, \ldots, i_{2 \log_2 q}), j])\,.\]
    \item Let $S' = \{j \in S \mid \alpha_j \text{ appears only once among } \{\alpha_j\}_{j \in S}\}$.
    Run the list-decoding algorithm from \cref{fact:GS} on $\{(\alpha_j,\beta_j)\}_{j \in S'}$.
    For each degree-$(m-1)$ polynomial in the output list, check whether at least $n' = O(n^{1-6\epsilon})$ pairs $\{(\alpha_j, \beta_j)\}_{j \in S'}$ satisfy $\beta_j = p(\alpha_j)$.
    If yes, return $1$.
    \item If the algorithm did not return $1$ on any guess, return $0$.
\end{enumerate}

The time complexity is dominated by the time to guess the indices.

\paragraph{Null case}
We first prove that, if $M \sim \cQ_n$, then the algorithm outputs $0$ with high probability.
The algorithm outputs $1$ only if there exists a $2 \log_2 q \times n'$ submatrix of $M$ (with distinct row and column indices) and a degree-$(m-1)$ polynomial such that $\log_2 q \cdot n'$ entries of the submatrix are a deterministic function (depending on the degree-$(m-1)$ polynomial) of the other $\log_2 q \cdot n'$ entries.
For a fixed submatrix and a fixed degree-$(m-1)$ polynomial, this event has probability $2^{-n' \log_2 q}$.
Then, for a fixed submatrix, by union bounding over all $q^m$ degree-$(m-1)$ polynomials, we get that the algorithm outputs $1$ with probability at most $2^{-(n'-m) \log_2 q}$.
Finally, we need to union bound over all submatrices of size $2 \log_2 q \times n'$. 
We note that the algorithm is invariant to the order of the columns in the submatrix, so it suffices to consider submatrices with ordered rows but unordered columns, of which there are at most $n^{2\log_2 q} \binom{n}{n'} \leq n^{2 \log_2 q} (en/n')^{n'} \leq n^{2 \log_2 q} \cdot O(n^{6\epsilon})^{n'}$.
Then we get overall that the algorithm outputs $1$ with probability at most $2^{-(n'-m) \log_2 q + 2 \log_2 n \log_2 q + O(\epsilon n') \log_2 n} \leq q^{-\Omega(n')}$.

\paragraph{Planted case}
We prove now that, if $M \sim \cP_n$, then the algorithm outputs $1$ with high probability.
Consider the guess in which $i_1, \ldots, i_{2 \log_2 q}$ correspond to the rows of the planted matrix $M_0$ from \cref{def:planted}.

We start by lower bounding the number of samples $\alpha_j$ from $M_0$ that appear only once among $\{\alpha_j\}_{j \in S}$ and whose bits are uncorrupted by the noise operator.
For some fixed $\alpha_j$ from $M_0$, the probability that no other $\{\alpha_j\}_{j \in S}$ is equal to it is at least $\Paren{1-1/q}^n \geq \Omega(1)$, and the probability that it is uncorrupted by the noise operator is at least $(1-\epsilon)^{\log_2 q}$.
These two events are independent, so the probability of both happening is at least 
$\Omega((1-\epsilon)^{\log_2 q})$.
There are $\floor{n/2}$ samples $\alpha_j$ in $M_0$, so from the above we get that the expected number of non-repeated and uncorrupted samples $\alpha_j$ from $M_0$ is at least $\Omega(n(1-\epsilon)^{\log_2 q})$.
Furthermore, if one of $\{\alpha_j\}_{j \in S}$ changes arbitrarily, the number of such non-repeated and uncorrupted samples can only increase or decrease by at most $3$.
Then, by McDiarmid's inequality, the probability that the number of non-repeated and uncorrupted samples $\alpha_j$ from $M_0$ is at least $\Omega(n(1-\epsilon)^{\log_2 q})$ is lower bounded by $1 - \exp\Paren{-\Omega\Paren{n(1-\epsilon)^{2\log_2 q}}} = 1-\exp\Paren{-n^{1-O(\epsilon)}}$.

Let us condition on the number of non-repeated and uncorrupted samples $\alpha_j$ from $M_0$ being at least $\Omega(n(1-\epsilon)^{\log_2 q})$.
The noise operator acts independently on the samples $\beta_j$ from $M_0$ corresponding to these $\alpha_j$.
For some fixed $\beta_j$ from $M_0$, the probability that its bits are uncorrupted by the noise is at least $(1-\epsilon)^{\log_2 q}$.
Then, out of at least $\Omega(n(1-\epsilon)^{\log_2 q})$ samples $\beta_j$ from $M_0$ corresponding to non-repeated and uncorrupted $\alpha_j$, by Binomial tail bounds we get that with probability $1 - \exp\Paren{-\Omega\Paren{n(1-\epsilon)^{3\log_2 q}}} = 1-\exp\Paren{-n^{1-O(\epsilon)}}$ at least $\Omega\Paren{n(1-\epsilon)^{2\log_2 q}}$ of them are uncorrupted by the noise.

Then, by the guarantees in \cref{fact:GS}, as long as $\Omega\Paren{n(1-\epsilon)^{2\log_2 q}} > \sqrt{nm}$, there are sufficiently many non-repeated and uncorrupted samples $(\alpha_j, \beta_j)$ from $M_0$ such that the list-decoding algorithm returns with high probability a list that includes the true polynomial relating these pairs, and the algorithm returns $1$.
The condition is satisfied for $m \leq n^{1-6\epsilon}$.
\end{proof}

\subsection{Extension to \texorpdfstring{$k > 2$}{k > 2}}
We now generalize our results to $k$-tensors with $k > 2$, for which we give a distinguisher with runtime $n^{O(k \log_2^{1/(k-1)} q)}$.

\begin{definition}[Null distribution $\cQ_n^{(k)}$] \label{def:nullk}
    The distribution $\cQ_n^{(k)}$ is the distribution over symmetric $k$-tensors $M \in \{0, 1\}^{n^{\otimes k}}$ with entries $M_{i_1,\ldots, i_k}$ with $i_1 < \ldots < i_k$ sampled i.i.d. from $\operatorname{Unif}(\{0, 1\})$.
\end{definition}

For $k \geq 2$ and $q \geq 2$ a power of two, we define the randomized function $\operatorname{binary}_{q}^{(k)} : \mathbb{F}_q \to \{0, 1\}^{\ceil{(\log_2 q)^{1/(k-1)}}^{\otimes (k-1)}}$ to map the $i$-th element of $\mathbb{F}_q$ to a canonical binary representation of $i$ as $\log_2 q$ binary entries in a $(k-1)$-tensor of dimension $\ceil{(\log_2 q)^{1/(k-1)}}^{\otimes (k-1)}$.
If $(\log_2 q)^{1/(k-1)}$ is not an integer and as a consequence the number of entries of the tensor is larger than $\log_2 q$, then the remaining binary entries are sampled independently from $\operatorname{Unif}(\{0, 1\})$.

\begin{definition}[Planted distribution $\cP_n^{(k)}$] \label{def:plantedk}
    For some $m \geq 2$ and $2 \leq q \leq 2^{\Omega(n)}$ a power of two, the distribution $\cP_n^{(k)}$ is sampled as follows:
    \begin{enumerate}[(1)]
        \item Sample $x_0, \ldots, x_{m-1} \stackrel{i.i.d}{\sim} \operatorname{Unif}(\mathbb{F}_q)$.
        \item Sample $\alpha_1, \ldots, \alpha_{\floor{n/2}} \stackrel{i.i.d.}{\sim} \operatorname{Unif}(\mathbb{F}_q)$ and define $\beta_1, \ldots, \beta_{\floor{n/2}} \in \mathbb{F}_q$ as $\beta_i = \sum_{j=0}^{m-1} x_j \alpha_i^j$ for all $i=1,\ldots,\floor{n/2}$.
        For any $\alpha_i$ whose value appears more than once among $\alpha_1, \ldots, \alpha_{\floor{n/2}}$, resample $\beta_i \sim \operatorname{Unif}(\mathbb{F}_q)$.
        \item Define $\ell = \ceil{(\log_2 q)^{1/(k-1)}}$, and let the $k$-tensor $M_0 \in \{0,1\}^{(2\ell)^{\otimes (k-1)} \times \floor{n / 2}}$ have:
        \begin{itemize}
            \item $T_0[1 : \ell,\; \ldots, 1: \ell, j] = \operatorname{binary}_q^{(k)}(\alpha_{j})$, for all $j=1,\ldots,\floor{n/2}$,

            \item $T_0[\ell+1 : 2\ell,\; \ldots, \ell+1: 2\ell, j] = \operatorname{binary}_q^{(k)}(\beta_{j})$, for all $j=1,\ldots,\floor{n/2}$.
        \end{itemize}
        \item Define $R_1 = [1 : 2\ell]$, $R_2 = [2\ell + 1 : 4\ell]$, ..., $R_{k-1} = [(2k-4)\ell + 1: (2k-2)\ell]$.
        Then let the symmetric $k$-tensor $M \in \{0, 1\}^{n^{\otimes k}}$ have $T[R_1, \ldots, R_{k-1}, \ceil{n/2}+1 : n] = M_0$, and let all its other entries $T_{i_1, \ldots, i_k}$ with $i_1 < \ldots < i_k$ be sampled i.i.d. from $\operatorname{Unif}(\{0, 1\})$.
        
        \item Apply a random $S_n$-permutation to the tensor $M$ and return it.
    \end{enumerate}
\end{definition}

The low-degree hardness for $\cQ_n^{(k)}$ and $\cP_n^{(k)}$ follows from the same argument as in \cref{lemma:rs-indist}.
It remains to prove that there exists an algorithm with runtime $n^{O(k \log_2^{1/(k-1)} q)}$ that distinguishes between the two distributions with high probability, even when noise is applied to the planted distribution.

\begin{lemma}
\label{lemma:rs-distk}
For constant $k \geq 2$, $q = \Theta(n)$ with $q \geq n$ a power of two, $\epsilon > 0$ a small enough constant, and $m \leq n^{1-6\epsilon}$, there exists an algorithm $\mathcal{A}$ that, given as input a symmetric tensor $M \in \{0, 1\}^{n^{\otimes k}}$, runs in time $n^{O(k \log_2^{1/(k-1)} q)}$ and satisfies
\[\Pr_{M \sim \cQ_n^{(k)}}(\mathcal{A}(M) = 0) \geq 1-\exp(-n^{1-O(\epsilon)})\,,\,\quad \Pr_{M \sim T_\epsilon\cP_n^{(k)}}(\mathcal{A}(M) = 1) \geq \exp(-n^{1-O(\epsilon)})\,.\]
\end{lemma}
\begin{proof}
The algorithm is:
\begin{enumerate}
    \item Define $\ell = \ceil{(\log_2 q)^{1/(k-1)}}$, and guess a list of size $2k-2$ whose elements are $\ell$-tuples of distinct ordered indices in $[n]$, and call these $\ell$-tuples $R_1, \ldots, R_{2k-2}$. Let the set of all indices guessed be $I$.
    \item For each $j \in \{1, \ldots, n\} \setminus I$, let
    \[\alpha_j = \Paren{\operatorname{binary}_q^{(k)}}^{-1}(M[R_1, \ldots, R_{k-1}, j])\,,\]
    \[\beta_j = \Paren{\operatorname{binary}_q^{(k)}}^{-1}(M[R_{k}, \ldots, R_{2k-2}, j])\,,\]
    where $\Paren{\operatorname{binary}_q^{(k)}}^{-1}$ is understood to ignore the redundant entries in its argument in the case that $(\log_2 q)^{1/(k-1)}$ is not an integer (see the definition of $\operatorname{binary}_q^{(k)}$ above).
    \item Run the list-decoding algorithm from \cref{fact:GS} on $\{(\alpha_j,\beta_j)\}_{j}$.
    For each degree-$(m-1)$ polynomial in the output list, check whether at least $n' = O(n^{1-6\epsilon})$ pairs $(\alpha_j, \beta_j)$ satisfy $\beta_j = p(\alpha_j)$.
    If yes, return $1$.
    \item If the algorithm did not return $1$ on any guess, return $0$.
\end{enumerate}

The time complexity is dominated by the time to guess the indices.
The rest of the analysis is analogous to that of \cref{lemma:rs-dist}.
\end{proof}

\section{Low-degree Conjecture vs the Top Eigenvalue}
\label{sec:top-eigenvalue}

In this section, we prove \Cref{thm:main-eigenvalue}.
\begin{theorem}
\label{thm:eigenvalue}
    There exist rotational invariant distributions $\cQ_n$ and $\cP_n$ over symmetric matrices in $\R^{n\times n}$ such that
    \begin{itemize}
        \item \textbf{Degree-$\poly(n)$ indistinguishable:} 
        $\LDLR{D}{\cP_n}{\cQ_n} \leq \frac{1}{\polylog(n)}$ for $D = n^{1/3}/\polylog(n)$.
        
        \item \textbf{Distinguishing algorithm after noise:} fix any $\eps \in [0,1)$.
        Let $M' = (1-\eps)M_1 + \eps M_0$ where $M_1 \sim \cP_n$ and $M_0 \sim \cQ_n$.
        Then, there is an algorithm $\calA$ that runs in polynomial time such that
        $\Pr_{M \sim \cQ_n}[\calA(M)=0] = 1$ and $\Pr_{M'}[\calA(M') = 1] \geq 1- \frac{1}{\poly(n)}$.
    \end{itemize}
\end{theorem}

The distributions are defined according to the following eigenvalue distribution.

\begin{definition}[Eigenvalue distribution $\mu_\gamma$] \label{def:eigenvalue-dist}
    Given parameter $\gamma \in (0,1)$, we define $\mu_{\gamma}$ to be the univariate distribution where $x \sim \Unif([-1,0])$ with probability $\gamma$ and $x=0$ otherwise.
\end{definition}

Next, we define the null and planted distributions according to $\mu_{\gamma}$.
Both are mostly supported on matrices of rank $\approx \gamma m \ll n$.
Moreover, matrices sampled from the null model are negative semidefinite, while those sampled from the planted model have exactly one positive eigenvalue (with high probability).

\begin{definition}[Null distribution $\cQ_n^{(\gamma,m)}$] \label{def:null-ev}
    Given parameters $\gamma \in (0,1)$ and $m, n\in \N$,
    the distribution $\cQ_n^{(\gamma, m)}$ is sampled as follows:
    \begin{enumerate}[(1)]
        \item Sample $\lambda_1,\lambda_2,\dots, \lambda_m \sim \mu_{\gamma}$ independently.
        \item Sample a random matrix $U \in \R^{n\times m}$ with i.i.d.\ $\calN(0,1)$ entries.
        \item Output $M = U \diag(\lambda)U^\top$.
    \end{enumerate}
\end{definition}

\begin{definition}[Planted distribution $\cP_n^{(\gamma,m,\lambda^*)}$] \label{def:planted-ev}
    Given parameters $\gamma \in (0,1)$, $m, n\in \N$, and $\lambda^* > 0$,
    the distribution $\cP_n$ is sampled as follows:
    \begin{enumerate}[(1)]
        \item Sample $\lambda_1,\lambda_2,\dots, \lambda_{m-1} \sim \mu_{\gamma}$ independently, and set $\lambda_m = \lambda^*$.
        \item Sample a random matrix $U \in \R^{n\times m}$ with i.i.d.\ $\calN(0,1)$ entries.
        \item Output $M = U \diag(\lambda)U^\top$.
    \end{enumerate}
\end{definition}

For simplicity of notation, we will drop the dependence on $\gamma, m, \lambda^*$ in the subsequent sections.
The two statements in \Cref{thm:eigenvalue} are proved in \Cref{lem:ev-ldlr,lem:ev-distinguisher} respectively.
The final proof (which is simply a combination of the two lemmas) are given in \Cref{sec:proof-of-ev-theorem}, where we set parameters $m = \Theta(n)$, $\gamma = \frac{\log^2 n}{n}$ and $\lambda^* = \gamma\log n$.

Our proofs also prove the following statement as an immediate corollary:

\begin{corollary}[Low-degree indistinguishability under noise]
\label{cor:null-null}
    Let $\cQ_n'$ be the distribution of $\frac{1}{2} M_0 + \frac{1}{2} M_0'$ where $M_0, M_0' \sim \cQ_n$, and let $\cP_n'$ be the distribution of $\frac{1}{2} M_0 + \frac{1}{2} M_1$ where $M_0 \sim \cQ_n$ and $M_1 \sim \cP_n$.
    Then, the two statements of \Cref{thm:eigenvalue} also hold for $\cQ_n'$ and $\cP_n'$.
\end{corollary}

\subsection{Efficient Distinguishing Algorithm}

We first need the well-known Hanson-Wright inequality \cite{MR279864-Hanson71, MR353419-Wright73} for the concentration of Gaussian quadratic forms (see also \cite{MR3125258-Rudelson13}).
\begin{fact}[Hanson-Wright inequality]
\label{fact:hanson-wright}
    Let $A \in \R^{n \times n}$ be a fixed matrix, and let $g\sim \calN(0, \Id_n)$.
    Then, there is a constant $c > 0$ such that for all $t > 0$,
    \begin{equation*}
        \Pr\bracks*{ \abs*{g^\top A g - \tr(A)} \geq t}
        \leq 2\exp\parens*{-c \cdot \min\parens*{\frac{t^2}{\|A\|_F^2}, \frac{t}{\|A\|_2}}} \mper
    \end{equation*}
\end{fact}

We now show that the top eigenvalue distinguishes between $\cQ_n$ and $\cP_n$ with high probability.

\begin{lemma}
\label{lem:ev-distinguisher}
    Fix $\eps \in [0,1)$.
    Let $M_0 \sim \cQ_n$ and $M_1 \sim \cP_n$ sampled independently, and let $M = (1-\eps)M_1 + \eps M_0$.
    Then, for $\gamma \geq \log^2 n/n$, we have $\lambda_1(M) \geq \Omega(\gamma n)$ with probability $1-\exp(-\wt{\Omega}(\gamma n))$.
\end{lemma}

\begin{proof}
    Let $M_1 = U \diag(\lambda) U^\top = \sum_{i=1}^n \lambda_i u_i u_i^\top$, where $\lambda_n = \gamma \log n$ and $\lambda_i \sim \mu$ for $i\leq n-1$.
    First, with probability $1 - \exp\parens{-\wt{\Omega}(n)}$, we have $\|u_i\|_2^2 \in (1\pm o(1)) n$ for all $i\in[n]$.
    Denote $k = \sum_{i=1}^{n-1} |\lambda_i|$ and $W = \{i \in [n-1]: \lambda_i \neq 0\}$.
    By \Cref{def:eigenvalue-dist}, $\E[k] = (n-1) \cdot \frac{1}{2}\gamma$ and $\E[|W|] = (n-1)\gamma$.
    Moreover, by the Chernoff bound, for any $\delta \in (0,1)$,
    \begin{align*}
        &\Pr\bracks*{k \notin (1\pm \delta) \gamma n /2} \leq 2 \exp\parens*{-\delta^2 \gamma n/6} \mcom \qquad 
        \Pr\bracks*{|W| \notin (1\pm \delta) \gamma n} \leq 2 \exp\parens*{-\delta^2 \gamma n/3} \mper
    \end{align*}
    Let $M' = \sum_{i=1}^{n-1} \lambda_i u_i u_i^\top$, which is negative semidefinite as $\mu$ is supported on $[-1,0]$.
    We have that $|\tr(M')| = \sum_{i=1}^{n-1} |\lambda_i| \cdot \|u_i\|_2^2 \leq (1+o(1))\gamma n^2$ and $M'$ has rank $|W| \leq (1+o(1))\gamma n$ with probability $1-\exp(-\wt{\Omega}(\gamma n))$.
    Moreover, conditioned on $\lambda$, $U_W = \{u_i\}_{i\in W}$ is an $n \times |W|$ random matrix with i.i.d.\ Gaussian entries.
    Thus, $\|U_W\|_2 \leq (1+o(1))\sqrt{n}$ with probability $1- \exp(-\wt{\Omega}(n))$, which means that $\|M'\|_2 \leq (1+o(1)) n$.
    In particular, this implies that $\|M'\|_F^2 \leq |W| \cdot \|M'\|_2^2 \leq 2\gamma n^3$.

    Applying the Hanson-Wright inequality (\Cref{fact:hanson-wright}), we have
    \begin{align*}
        \Pr\bracks*{|u_n^\top M' u_n| \geq |\tr(M')| + t} \leq 2\exp\parens*{-c \cdot \min\parens*{\frac{t^2}{\gamma n^3},\ \frac{t}{n} }} \mcom
    \end{align*}
    for some universal constant $c$.
    Setting $t = \gamma n^2/\log n$, it follows that $|u_n^\top M' u_n| \leq (1+o(1))\gamma n^2$ with probability at least $1 - \exp(-\wt{\Omega}(\gamma n))$.

    For $M_0 \sim \cQ_n$, the same calculation shows that $|u_n^\top M_0 u_n| \leq (1+o(1)) \gamma n^2$.
    On the other hand, $\lambda_n \|u_n\|_2^4 \geq  \gamma \log n \cdot (1-o(1))n^2$.
    Thus, for $M = (1-\eps) M_1 + \eps M_0$,
    \begin{align*}
        u_n^\top M u_n
        &\geq (1-\eps) \lambda_n \|u_n\|_2^4 - (1-\eps) |u_n^\top M'u_n| - \eps |u_n^\top M_0 u_n| \\
        &\geq (1-\eps) \cdot (1-o(1)) \gamma n^2 \log n - (1+o(1)) \gamma n^2
        > 0 \mper
    \end{align*}
    Thus, $M$ has a positive eigenvalue with probability at least $1- \exp(-\wt{\Omega}(\gamma n))$.
\end{proof}

\subsection{Low-Degree Indistinguishability}

We first show that we only need to consider low-degree \emph{symmetric} polynomials of the (approximate) eigenvalues.

\begin{lemma} \label{lem:symmetric-poly}
    Let $p$ be a degree-$d$ polynomial in $n^2$ variables, and let $U \in \R^{n \times m}$ be a random matrix with i.i.d.\ $\calN(0,1)$ entries.
    Then, the polynomial $q: \R^{m} \to \R$ defined as $q(\lambda) \coloneqq \E_U[p(U \diag(\lambda) U^\top)]$ has degree $d$ and is a symmetric polynomial in $m$ variables.
\end{lemma}

\begin{proof}
    It is clear that $q$ has degree $d$, thus it suffices to prove that $q$ is symmetric.
    We start by writing $p$ in the monomial basis:
    \begin{equation*}
        p(M) = \sum_{t=0}^d \angles*{C^{(t)}, M^{\otimes t}} \mcom
    \end{equation*}
    where $C^{(t)} \in (\R^{n\times n})^{t}$ are the coefficient tensors.
    Take $M = U \diag(\lambda) U^\top = \sum_{i=1}^{m} \lambda_i u_i u_i^\top$, where $u_1,\dots,u_{m}$ are the columns of $U$.
    Then, for any $t \leq d$,
    \begin{align*}
        \E_U\bracks*{ \angles*{C^{(t)}, M^{\otimes t}}}
        &= \angles*{C^{(t)}, \E_U \parens*{\sum_{i=1}^{m} \lambda_i u_iu_i^\top}^{\otimes t}} \\
        &= \sum_{i_1,\dots,i_t \in [m]} \lambda_{i_1} \cdots \lambda_{i_t} \angles*{C^{(t)}, \E_U \bracks*{(u_{i_1} u_{i_1}^\top) \otimes \cdots \otimes (u_{i_t} u_{i_t}^\top) }} \mper
    \end{align*}
    Observe that since $U$ has i.i.d.\ entries, $\E_U \bracks{(u_{i_1} u_{i_1}^\top) \otimes \cdots \otimes (u_{i_t} u_{i_t}^\top)}$ only depends on the repeating pattern of $(i_1,\dots,i_t)$.
    More specifically, for any permutation $\pi \in \calS_{m}$, we have
    \begin{equation*}
        \E_U \bracks*{(u_{i_1} u_{i_1}^\top) \otimes \cdots \otimes (u_{i_t} u_{i_t}^\top)}
        = \E_U \bracks*{(u_{\pi(i_1)} u_{\pi(i_1)}^\top) \otimes \cdots \otimes (u_{\pi(i_t)} u_{\pi(i_t)}^\top)} \mper
    \end{equation*}
    Thus, for any $\pi \in \calS_{m}$, we have
    \begin{align*}
        q(\lambda_{\pi(1)},\dots, \lambda_{\pi(m)})
        &= \sum_{i_1,\dots,i_t \in [m]} \lambda_{\pi(i_1)} \cdots \lambda_{\pi(i_t)} \angles*{C^{(t)}, \E_U \bracks*{(u_{i_1} u_{i_1}^\top) \otimes \cdots \otimes (u_{i_t} u_{i_t}^\top) }} \\
        &= \sum_{i_1,\dots,i_t \in [m]} \lambda_{\pi(i_1)} \cdots \lambda_{\pi(i_t)} \angles*{C^{(t)}, \E_U \bracks*{(u_{\pi(i_1)} u_{\pi(i_1)}^\top) \otimes \cdots \otimes (u_{\pi(i_t)} u_{\pi(i_t)}^\top)}} \\
        &= q(\lambda) \mcom
    \end{align*}
    which proves that $q$ is symmetric.
\end{proof}

The next lemma shows that given our null and planted models (\Cref{def:null-ev,def:planted-ev}), we can further assume that the polynomial is of the form $\sum_{i=1}^m q(\lambda_i)$.

\begin{lemma} \label{lem:trace-poly}
    Let $\nu_{\cQ} = \mu^m$ and $\nu_{\cP} = \mu^{m-1} \times \delta_{\lambda^*}$ (as defined in \Cref{def:planted-ev}).
    For any degree-$d$ polynomial $p$ in $n^2$ variables with $\E_{M\sim \cQ_n}[p(M)] = 0$, there is a degree-$d$ univariate polynomial $q$ such that
    \begin{enumerate}[(1)]
        \item $\E_{\lambda\sim\mu}[q(\lambda)] = 0$.
        \item $\E_{M\sim \cQ_n}[p(M)^2] \geq \E_{\lambda \sim \nu_{\cQ}}[(\sum_{i=1}^m q(\lambda_i))^2]$.
        \item $\E_{M\sim \cP_n}[p(M)] = \E_{\lambda \sim \nu_{\cP}} [\sum_{i=1}^m q(\lambda_i)]$.
    \end{enumerate}
\end{lemma}
\begin{proof}
    For both $\cQ_n$ and $\cP_n$, the matrix is sampled to be $U \diag(\lambda) U^\top$ where $U \in \R^{n\times m}$ is a random Gaussian matrix and $\lambda \in \R^m$ is sampled from either $\nu_{\cQ}$ or $\nu_{\cP}$.
    Thus, by \Cref{lem:symmetric-poly}, we may consider the degree-$d$ $m$-variate symmetric polynomial $f(\lambda) = \E_U[p(U\diag(\lambda)U^\top)]$, where $f(\lambda)$ and $p(M)$ have the same expectation under both $\cQ_n$ and $\cP_n$,
    and $\E_{\lambda\sim\nu_{\cQ}}[f(\lambda)^2] \leq \E_{M\sim \cQ_n}[p(M)^2]$ by Jensen's inequality.

    We next show that further restricting $f(\lambda)$ to some polynomial of the form $\sum_{i=1}^m q(\lambda_i)$ decreases the variance.
    The distribution $\mu$ defines an inner product $\angles{f,g}_\mu = \E_{x\sim \mu}[f(x)g(x)]$.
    We can perform the Gram-Schmidt process on the monomials $1,x,x^2,\dots$ to obtain an orthonormal basis $\{\psi_i\}_{i\in \N}$ such that
    \begin{itemize}
        \item $\psi_i$ is a polynomial of degree $i$,
        \item $\E_{x\sim\mu}[\psi_i(x)] = 0$ for $i\geq 1$,
        \item $\E_{x\sim \mu}[\psi_i(x) \psi_j(x)] = \1(i=j)$ for all $i,j\geq 0$.
    \end{itemize}
    For example, we have $\psi_0(x)=1$, $\psi_1(x) = c(x- \E_{x\sim\mu}[x])$ (where $c$ is a normalizing constant), and so on.
    Then, the polynomial $f$ can be written as a linear combination of $\prod_{i=1}^m \psi_{\alpha_i}(\lambda_i)$ where $\alpha\in \N^m$ and $\|\alpha\|_1 \leq d$.
    Moreover, since $f$ is symmetric (in $m$ variables), it can be expressed as
    \begin{equation*}
        f(\lambda) = \sum_{\alpha \in A_d} c_\alpha \E_{\pi\sim \calS_m} \bracks*{\prod_{i=1}^m \psi_{\alpha_i}(\lambda_{\pi(i)})} \mcom
    \end{equation*}
    where $A_d \coloneqq \{\alpha\in \N^m: \alpha_1 \geq \alpha_2 \geq \cdots \geq \alpha_m \geq 0, \|\alpha\|_1\leq d\}$.
    
    First, we have $\E_{\lambda \sim \nu_{\cQ}}[f(\lambda)] = c_0 = 0$.
    Moreover, for any $\alpha,\beta \in A_d$, we have
    \begin{equation*}
        \E_{\lambda\sim\nu_{\cQ}} \E_{\pi\sim \calS_m} \E_{\pi'\sim \calS_m} \prod_{i=1}^m \psi_{\alpha_i}(\lambda_{\pi(i)}) \psi_{\beta_i}(\lambda_{\pi'(i)})
        = \begin{cases}
            r_\alpha & \alpha = \beta \mcom \\
            0 & \alpha \neq \beta \mcom
        \end{cases}
    \end{equation*}
    for some $r_\alpha > 0$.
    Therefore, it follows that
    \begin{equation*}
        \E_{\lambda \sim \nu_{\cQ}}[f(\lambda)^2] = \sum_{\alpha\in A_d} c_\alpha^2 \cdot r_\alpha \mper
    \end{equation*}
    On the other hand, for $\nu_{\cP}$, $\E_{\lambda\sim \nu_{\cP}} \prod_{i=1}^m \psi_{\alpha_i}(\lambda_{\pi(i)})$ is nonzero only when $\alpha = (k,0,\dots,0)$ for some $k\in \N$ and $\pi(1) = m$.
    Thus,
    \begin{equation*}
        \E_{\lambda \sim \mu_P}[f(\lambda)] = \sum_{k=1}^d c_{(k,0,\dots,0)} \cdot \frac{1}{m}\psi_k(\lambda^*) \mper
    \end{equation*}
    Thus, denote $b_k \coloneqq c_{(k,0,\dots,0)}$ (with $b_0 = 0$) and define the polynomial $g$ to be
    \begin{equation*}
        g(\lambda) \coloneqq \sum_{k=1}^d b_k \E_{\pi\sim\calS_m}[\psi_k(\lambda_{\pi(1)})] = \frac{1}{m}\sum_{i=1}^m \sum_{k=1}^d b_k \psi_k(\lambda_i) \mcom
    \end{equation*}
    then we have $\E_{\lambda\sim\nu_{\cQ}}[f(\lambda)] = \E_{\lambda\sim\nu_{\cQ}}[g(\lambda)] = 0$ and
    $\E_{\lambda\sim\nu_{\cQ}}[g(\lambda)^2] = \sum_{k=0}^d b_k^2 \cdot r_{(k,0,\dots,0)} \leq \E_{\lambda\sim\nu_{\cQ}}[f(\lambda)^2]$,
    and moreover $\E_{\lambda\sim\nu_{\cP}}[g(\lambda)] = \E_{\lambda\sim\nu_{\cP}}[f(\lambda)] = \frac{1}{m}\sum_{k=1}^d b_k \psi_k(\lambda^*)$.

    Now, define $q$ to be the following degree-$d$ univariate polynomial:
    \begin{equation*}
        q(x) \coloneqq \frac{1}{m}\sum_{k=1}^d b_k \psi_k(x) \mper
    \end{equation*}
    Observe that $g(\lambda) = \sum_{i=1}^n q(\lambda_i)$.
    It follows that $q$ satisfies all $3$ statements of the lemma, completing the proof.
\end{proof}

\paragraph{Legendre polynomials}
The univariate Legendre polynomials $\{L_k\}_{k\in \N}$ are defined by the following recurrence:
\begin{equation*}
    L_0(x) = 1,\quad L_1(x) = x, \quad (k+1)L_{k+1}(x) = (2k+1)xL_k(x) - k L_{k-1}(x) \mcom
\end{equation*}
and they have the following explicit expressions:
\begin{equation*}
    L_k(x) = \sum_{i=0}^k \binom{k}{i} \binom{k+i}{i} \parens*{\frac{x-1}{2}}^i \mper
\end{equation*}
An important property that one can verify is that $L_k(1) = 1$ for all $k\in \N$.
Moreover, the polynomials are orthogonal with respect to the uniform distribution over $[-1,1]$:
\begin{equation*}
    \E_{x\sim \Unif([-1,1])}[L_k(x) L_\ell(x)] = \frac{1}{2k+1} \delta_{k\ell} \mcom
\end{equation*}
where $\delta_{k\ell} = 1$ if $k=\ell$ and $0$ otherwise.

For our convenience, we define the following shifted polynomial:
\begin{align*}
    \wt{L}_k(x) = L_k(2x+1) = \sum_{i=0}^k \binom{k}{i} \binom{k+i}{i} x^i \mper
    \numberthis \label{eq:legendre-explicit}
\end{align*}
which are orthogonal with respect to $\Unif([-1,0])$.

\paragraph{Low-degree indistinguishability}
We now prove the statement that the null and planted distributions are low-degree indistinguishable.

\begin{lemma}
\label{lem:ev-ldlr}
    Let $n,m,d\in \N$ and $\gamma, \lambda^* > 0$ be such that
    $m = \Theta(n)$, $\lambda^* = \gamma \log n \leq \frac{1}{2d(d+1)}$ and $d^3 \sqrt{\gamma \log n/n} \leq o(1)$.
    Then, for any degree-$d$ polynomial $p$ in $n^2$ variables such that $\E_{M\sim \cQ_n}[p(M)] = 0$ and $\E_{M\sim \cQ_n}[p(M)^2] \leq 1$, we have $\E_{M\sim \cP_n}[p(M)] \leq O(d^3 \sqrt{\gamma \log n/n}) + O(n^{-1/2})$.
\end{lemma}
\begin{proof}
    By \Cref{lem:trace-poly}, we only need to consider polynomials of the form $\sum_{i=1}^m q(\lambda_i)$, where $q$ is a univariate polynomial of degree $d$, and $\lambda$ is sampled from either $\mu^m$ or $\mu^{m-1} \times \delta_{\lambda^*}$.

    First, we write $q$ in terms of the shifted Legendre polynomials (\Cref{eq:legendre-explicit}): $q(x) = \sum_{k=0}^d c_k \wt{L}_k(x)$.
    For the null model $\cQ_n$, we have $\E_{M\sim \cQ_n}[p(M)] = 0$ implies that $\E_{x\sim \mu}[q(x)] = 0$.
    Next, we have $1 \geq \E_{M\sim \cQ_n}[p(M)^2] \geq \E_{\lambda \sim \mu^m}[(\sum_{i=1}^m q(\lambda_i))^2] = m \cdot \E_{x\sim\mu}[q(x)^2]$, and
    \begin{align*}
        \E_{x\sim \mu}[q(x)^2]
        &= (1-\gamma)\cdot q(0)^2 + \gamma \cdot \E_{x\sim \Unif([-1,0])}[q(x)^2] \\
        &= (1-\gamma) \cdot q(0)^2 + \gamma \sum_{k=0}^d c_k^2 \cdot \frac{1}{2k+1} \mcom
    \end{align*}
    where the last equality follows from the orthogonality of $\wt{L}_k$ under $\Unif([-1,0])$.
    Thus,
    \begin{equation*}
        (1-\gamma)q(0)^2 + \gamma \sum_{k=0}^d \frac{c_k^2}{2k+1} \leq \frac{1}{m} \mper
        \numberthis \label{eq:squared-leq-1/n}
    \end{equation*}
    This implies that $|q(0)| \leq \frac{1}{\sqrt{(1-\gamma)m}}$ and $\sum_{k=0}^d \frac{c_k^2}{2k+1} \leq \frac{1}{\gamma m}$.

    Next, for the planted model $\cP_n$, we have $\lambda_1,\dots, \lambda_{m-1} \sim \mu$ and $\lambda_m = \lambda^*$.
    Let $\nu_{\cP} \coloneqq \mu^{m-1} \times \delta_{\lambda^*}$.
    Then, 
    \begin{align*}
        \E_{\lambda \sim \nu_{\cP}}\sum_{i=1}^m q(\lambda_i) = \sum_{i=1}^{m-1} \E_{\lambda\sim\mu}[q(\lambda_i)] + q(\lambda^*)
        = q(\lambda^*) = \sum_{k=0}^d c_k \wt{L}_k(\lambda^*) \mper
    \end{align*}
    By \Cref{eq:legendre-explicit},
    \begin{equation*}
        \wt{L}_k(\lambda^*) = \sum_{i=0}^k \binom{k}{i} \binom{k+i}{i} (\lambda^*)^i
        \leq 1 + \sum_{i=1}^k \parens*{k(k+1)\lambda^*}^i
        \leq \wt{L}_k(0) + 2k(k+1)\lambda^* \mcom
    \end{equation*}
    as long as $\lambda^* \leq \frac{1}{2k(k+1)}$ (for $k\geq 1$).
    Here, we also use that $\wt{L}_k(0) = L_k(1) = 1$.
    Thus,
    \begin{align*}
        q(\lambda^*) &\leq q(0) + \lambda^* \sum_{k=0}^d c_k \cdot k(k+1)
        \leq q(0) + \lambda^* \sqrt{\sum_{k=0}^d \frac{c_k^2}{2k+1}} \sqrt{\sum_{k=0}^d (2k+1)k^2(k+1)^2} \\
        &\leq \frac{1}{\sqrt{(1-\gamma) m}} + \lambda^*  \sqrt{\frac{1}{\gamma m}} \cdot O(d^3) \mcom
    \end{align*}
    where the last inequality follows from \Cref{eq:squared-leq-1/n}.
    Suppose $\lambda^* = \gamma \log n$ and $m = \Theta(n)$,
    then the above is at most $O(n^{-1/2}) + O(d^3\sqrt{\gamma \log n / n})$.
    This completes the proof.
\end{proof}

\subsection{Finishing the Proof}
\label{sec:proof-of-ev-theorem}

We prove \Cref{thm:eigenvalue} by combining \Cref{lem:ev-ldlr,lem:ev-distinguisher}.

\begin{proof}[Proof of \Cref{thm:eigenvalue}]
    For our null and planted distributions (\Cref{def:null-ev,def:planted-ev}), we set
    $m = \Theta(n)$, $\gamma = \frac{C\log^2 n}{n}$ and $\lambda^* = \gamma\log n$ for some large constant $C > 1$.
    
    The distinguishing algorithm simply checks whether the input matrix has a positive eigenvalue or not.
    For $M \sim \cQ_n$, $M$ is negative semidefinite (with probability $1$).
    For $M \sim \cP_n$, by \Cref{lem:ev-distinguisher}, $\lambda_1(M) > \Omega(\gamma n)$ with probability at least $1 - \frac{1}{\poly(n)}$.
    
    On the other hand, let $D = n^{1/3}/\polylog(n)$, which satisfies the conditions $\lambda^* \leq o(D^{-2})$ and $D^3 \sqrt{\gamma \log n /n} \leq 1/\polylog(n)$ required in \Cref{lem:ev-ldlr}.
    We have that $\LDLR{D}{\cP_n}{\cQ_n} \leq 1/\polylog(n)$.
\end{proof}

\begin{proof}[Proof of \Cref{cor:null-null}]
    The proof follows by observing that $M_0 + M_0'$ is distributed as $\cQ_n^{(\gamma,2m)}$ while $M_0 + M_1$ is distributed as $\cP_n^{(\gamma,2m,\lambda^*)}$.
\end{proof}

\section*{Acknowledgments}
The last-named author thanks Afonso Bandeira for insightful discussions on a previous attempt to disprove the low-degree conjecture, and Alex Wein for detailed comments on a previous version of this paper. We also thank Tselil Schramm and Sam Hopkins for insightful comments on this work, and Jeff Xu for related discussions over the past year. We also thank the researchers at the AIM workshop on low-degree polynomial methods in average-case complexity for insightful comments on the second counterexample presented in this work.

\bibliographystyle{alpha}
\bibliography{main}

\end{document}